\documentclass[a4paper]{llncs}
\usepackage{amsmath}

\usepackage{amssymb}
\usepackage{multirow}
\usepackage{graphicx}
\usepackage{mathtools}

\DeclarePairedDelimiter\floor{\lfloor}{\rfloor}
\usepackage{listings}
\usepackage{chngcntr}
\usepackage{url}

\spnewtheorem{mainth}[theorem]{Main Theorem}{\bfseries}{\itshape}
\newtheorem{dfn}[theorem]{Definition}

\newcommand{\RR}{\mathbb{R}}
\newcommand{\e}{\epsilon}
\newcommand{\Pe}{P^{-\e}}
\newtheorem{cor}[theorem]{Corollary}
\newtheorem{lem}[theorem]{Lemma}
\def\RR{{\mathbb R}}

\begin{document}
\title{Algorithms for Deciding Membership in Polytopes of General Dimension}

\author{Evangelos Anagnostopoulos\inst{1}, Ioannis Z.~Emiris\inst{1}\and Vissarion Fisikopoulos\inst{2}}
\institute{National and Kapodistrian University of Athens, Athens, Greece, \email{\{aneva,emiris\}@di.uoa.gr}\\\and Oracle Corp., Greece, \email{vissarion.fysikopoulos@oracle.com}}

\maketitle

\begin{abstract}
We study the fundamental problem of polytope membership aiming at large convex polytopes, i.e.\ in high dimension and with many facets, given as an intersection of halfspaces.
Standard data-structures as well as brute force methods cannot scale, due to the curse of dimensionality. 
We design an efficient algorithm, by reduction to the approximate Nearest Neighbor (ANN) problem based on the construction of a Voronoi diagram with the polytope being one bounded cell.
We thus trade exactness for efficiency so as to obtain complexity bounds polynomial in the dimension, by exploiting recent progress in the complexity of ANN search. We employ this algorithm to present a novel boundary data structure based on a Newton-like iterative intersection procedure.
We implement our algorithms and compare with brute-force approaches to show that they scale very well as the dimension and number of facets grow larger.  

\end{abstract} 

\section{Introduction}

In geometric optimization, convex polytopes are very important objects appearing also as feasible regions in linear programming.
Let us consider a convex polytope $P$ in H-representation, that is as the intersection of a finite set of linear inequalities:
$P = \{x\in\RR^d \mid Ax\leq b,\, A\in\RR^{n \times d},\, b\in\RR^n \}$.
An important question on such a polytope is that of point membership.
We wish to preprocess $P$ in order to obtain a membership data structure which, given a query point $q$, efficiently decides whether $q$ lies inside or outside $P$. A decision can be reached by testing all $n$ inequalities for a complexity of $O(nd)$. This trivial approach is often a plausible exact solution, especially in the high-dimensional case. In order to design a more efficient algorithm in high dimension, we will focus on the approximate polytope membership problem where the membership data structure is allowed to answer incorrectly for points lying very close to the boundary of the polytope. A formal definition will be provided later in Section~\ref{sAPMO}.

Algorithms used to solve combinatorial optimization problems, such as the ellipsoid, interior point or randomized methods (for the latter see~\cite{Bertsimas04}), usually rely on randomly sampling convex polytopes. 
The inner loop of such algorithms needs access to a membership or a boundary oracle, where the latter is the procedure that computes the intersection of a ray with the boundary of the polytope and is equivalent to membership via binary search. The oracle specification means that we are not interested in how the solution is computed or of its computational complexity.
Gr\"otschel et al.~\cite{Groetschel88} proposed the oracle model of computation and among other results they prove the polynomial time equivalence of basic oracles such as optimization, separation, and membership. This has become a commonly employed tool in combinatorial optimization mainly for studying the computational complexity of problems. Another important example of application is volume approximation~\cite{Dyer91,Lovasz06} which has also an established connection to combinatorial optimization. For example, the volume of order polytopes gives the number of linear extentions of the associated partial order set. 

From a practical point of view opening the oracle black box, in particular membership, and improving their complexity, implies improvements to the applicability of the aforementioned algorithms. For example, the first implementation of randomized algorithms that scale in high dimension appeared in~\cite{EmiFis14}.
Their approach relies on the standard random walks known as hit-and-run, which require a boundary oracle.
Notice that, although this software can handle polytopes in spaces whose dimension goes up to $200$, it cannot scale as efficiently for specific classes of polytopes with a large number of facets.
In particular, it cannot approximate the volume of cross-polytopes of dimension~$20$ or more.

Here, we radically shift the aforementioned paradigm and, moreover, improve upon the complexity of membership and boundary data structures, when dimension $d$ is an input parameter. We exploit the approximate setting and allow ourselves to answer correctly within some approximation error $\epsilon$ and with some success probability.
Our new paradigm uses a reduction to the Approximate Nearest Neighbor (ANN) problem, which is the most fundamental problem among those today with a practical, poly-time solution in high-dimensions.

\paragraph{Previous Work.} 
There are two classical results for the approximate membership problem, both based on creating $\epsilon$-approximating polytopes and answering membership on them.
Any convex body is $\epsilon$-approximated by a polytope with $O(1/\epsilon^{(d-1)/2})$ facets, which is asymptotically tight in the worst case~\cite{Dudley74}.
This leads to a membership data structure with space and query complexity in $O(1/\epsilon^{(d-1)/2} )$. 
Using a $d$-dimensional grid, membership takes constant time (assuming a model of computation that supports the floor function) and space grows to $O(1/\epsilon^{d-1})$~\cite{Bentley82}.

A relevant line of work on approximate membership in fixed $d$ uses space-time trade-offs~\cite{Arya11,Arya12soda} to
achieve a space of $O(1/\epsilon^{(d-1)\left(1-(2\floor{\log t}-2)/t\right)})$ with query time $O(\log({1/\epsilon})/\epsilon^{(d-1)/t})$, for trade-off parameter $t \geq 4$.
In~\cite{Arya17}, again for fixed $d$, they opt for a hierarchy of ellipsoids selected by a sampling process on classical structures from the theory of convexity defined on the polytope. They achieve space $O(1/\epsilon^{(d-1)/2})$ with an optimal query time of $\log({1/\epsilon})$.

We present state-of-the-art approaches to ANN as we build atop of those for our oracles. There are many solutions to this problem, but in principle, methods that scale polynomially with $d$ belong to two categories.  First, the well studied Locality Sensitive Hashing (LSH)~\cite{Indyk98}. The other category focuses on random projections~\cite{Anag15}, then uses fast algorithms in fixed dimension. 
Both achieve sublinear query time with (near-)linear storage, while scaling polynomially in $d$, and both have a probability of success $p$. 

\paragraph{Our contribution.} 
We describe a simple constructive reduction from the polytope membership 
problem to ANN, then show under which conditions this 
reduction holds for the respective approximate versions of the problems. This gives 
us the flexibility to exploit advances in the research of ANN in order to offer, the first (as far as the authors are aware) practical approximate polytope membership data structure in high dimension with complexity bounds polynomial in the dimension $d$ and sublinear in the number of inequalities $n$. This is our main result, in Theorem~\ref{mainTheorem}.
We also present an application of this membership data structure for creating boundary data structures for H-polytopes. 
We implement and experimentally examine our algorithms; we illustrate that they
scale well as dimension and number of facets grow larger.  
Our implementation is linked to the software of~\cite{EmiFis14} for polytope volume, so as to provide faster oracles.

The rest of the paper is organized as follows.
The next section discusses (approximate) membership and the reduction to ANN.
Section \ref{Sbound} considers the boundary data structures.
The implementation and experiments are in Section \ref{Simplement}.
We conclude with open questions.

\section{Approximate Polytope Membership}
We assume that the given H-polytope $P$ is full dimensional and that its representation is minimal, i.e. that it does not contain redundant inequalities.

We denote the $i$-th (in)equality of $P$ as $a_i x \leq b_i, 1 \leq i \leq n$. We associate each facet of the polytope with a corresponding (in)equality and denote it as $F_i$. Formally: $ F_i = \{x \in P \mid a_i x = b_i \},\ 1 \leq i \leq n $. 
The hyperplanes that define non-empty $F_i$'s, i.e.\ for which $F_i \neq \emptyset$ are called non-redundant or supporting and we extend that label to their inequalities. We denote as $\partial P$ the boundary of $P$:
$\partial P = \{x \in P \mid \exists i, \ \ 1 \leq i \leq n \ \ \text{s.t.} \ \ x \in F_i \}$.

\subsection{Exact Polytope Membership Oracle}
A reduction from the exact polytope membership problem to the exact nearest neighbor problem was established in \cite{Aurenhammer87}, where it was shown that there is a connection between the boundaries of polytopes in $\RR^d$ and power diagrams in $\RR^{d-1}$. Power diagrams define a partition of the Euclidean space into a cell complex based on a set of spheres. Each sphere identifies a specific cell and that cell consists of all the points whose power distance is minimized for that sphere. The power diagram is a generalized Voronoi diagram, and coincides with the Voronoi diagram of the sphere centers if all spheres have equal radii. 


\begin{theorem} {\cite[Thm.4]{Aurenhammer87}}
For any polyhedron $P\in\RR^d$, which is expressible as the intersection of upper halfspaces, there exists an affinely equivalent power diagram in hyperplane $h_0: x_d=0$.
\end{theorem}

\noindent A cell complex $C$ and a polyhedron $P \subset \RR^{d+1}$ are said to be affinely equivalent if there exists a central or parallel projection $\phi$ such that, for each face $f$ of $C,\ f= \phi(g)$ holds for some face $g$ of $P$. This provides a reduction from ray shooting in a polyhedron to point location in a polyhedral complex. 
In the case of polytope membership, the polyhedral complex becomes a single cell (the polytope) and the power diagram becomes a Voronoi diagram. This provides a reduction from polytope membership to Nearest neighbor.

\begin{cor}\label{Cpoints}
Let $P \subset \mathbb{R}^d$ be a convex polytope described as the intersection of $n$ non-redundant
halfspaces. For every point $p^* \in P \setminus \partial P$ it is possible to compute a set $S$ of $n+1$ points such that, $p^* \in S$ and, given a query point $q$, the exact
Polytope Membership test for a query point $q$ reduces to finding the Nearest Neighbor of $q$ among these $n+1$ points. 
\end{cor}

\begin{proof}

We initialize $S = \{p^*\}$.  
We will describe for completeness the procedure to compute the remaining $n$ points of $S$ such that the corresponding Voronoi diagram of these $n$ points and $p^*$ will have the polytope $P$ as the voronoi cell of $p^*$. These $n+1$ points will be the points of the corollary.

For each facet $F_i$ and its corresponding hyperplane $H_i := a_i x = b_i,\, 1 \leq i \leq n$, we compute the projection of $p^*$ on $H_i$ and denote it as $f_i$.  
Then, we compute the point $p_i,\, 1\le i\le n$, such that the line segment $(p^*,p)$ is perpendicular to $H_i$ and $d(p^*, H_i) = ||p^* - f_i||_2 = d(p_i, H_i)$, where $d(p, S) = \min\limits_{x \in S} ||p-x||_2$.
Equivalently, 
$
p_i = f_i + (f_i - p^*)
$.

We now have a set of points $S = \{p^*, p_1, \ldots, p_n\}$ of $n+1$ points that have the following property.
In the Voronoi diagram of $S$, by construction, the cell that corresponds to $p^*$ is 
precisely the input polytope $P$.  
By the Voronoi property, the following holds:
$
q \in P \Leftrightarrow ||p^*- q||_2 \leq ||q - s||_2, \; \forall s \in S.
$
Polytope membership returns ``YES" iff the nearest neighbor of $q$ is $p^*$. \qed
\end{proof}

\begin{figure}[t]\centering
\includegraphics[scale=0.60]{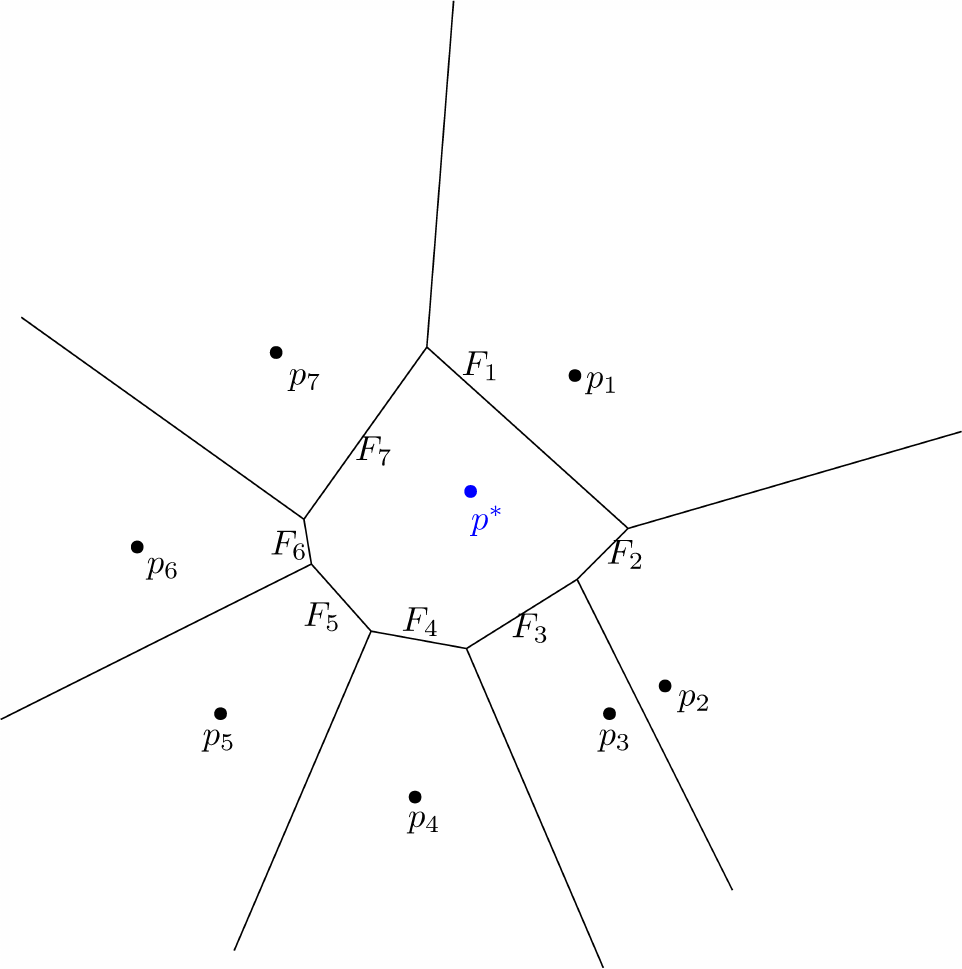}
\caption{A conceptual presentation of the constructive proof in the case of d=2. Each $p_i$ corresponds to the symmetric point of $p^*$ about the facet $F_i$. 
\label{fVoronoiExample}}
\end{figure}

\textit{Remark. } A nearest neighbor computation or data structure on these $n+1$ points of corollary~\ref{Cpoints} provides us with an exact Membership Oracle for the polytope $P$. We also emphasize that the choice of $p^* \in P$ is arbitrary. This means that a set $S$ satisfying the Corollary can be computed for each point $p^* \in P \setminus \partial P$.

\subsection{Approximate Polytope Membership Oracle}\label{sAPMO}

Let us consider the following relaxation.

\begin{dfn}[Approximate Polytope Membership Problem]\label{dfnAPMP}
Given a convex polytope $P\subset\RR^d$ and an approximation parameter $\e\in (0,1)$, an $\e$-approximate polytope membership query decides whether a query point $q\in\RR^d$ lies inside or outside of $P$, but may return either answer if $q$'s distance from the boundary of $P$ is at most $\e \cdot diam(P)$.
\end{dfn}

We define $\Pe = \{x \in P \mid d(x, \partial P) > \e \cdot diam(P) \}$.
 Obviously the aforementioned problem makes sense only when $\Pe \neq \emptyset$. Otherwise, we can always return ``NO'' for a query point $q$ and be correct.

\begin{theorem}[Approximate Membership Oracle (AMO)]\label{lemmaAPM}
Approximate Polytope Membership for an $H$-polytope $P$ and an approximation parameter $\e$, such that $\Pe \neq \emptyset$, reduces to the ANN problem on the pointset $S = \{p^*,\, p_i : 1\le i \le n\}$, where $p^* \in \Pe$ and the remaining $p_i$ are computed as in the proof of Corollary~\ref{Cpoints}.  
\end{theorem}

\begin{proof}
Let $p^* \in \Pe$ and $S$ be the corresponding pointset of Lemma~\ref{Cpoints} for $P$. Let $\Delta(P) = \max\limits_{p_i \in S \setminus \{p^*\}} ||p_i-p^*||_2 $. By construction, the following holds for $\Delta(P)$:
$
2\e \cdot diam(P) < \Delta(P) < 2 diam(P)
$.
Let $q \in \RR^d$ be a query point such that $||q - p^*|| < \frac{\Delta(P)}{2\e}$. For any other $q' \in \RR^d$, we return ``NO'', because 
$
||q' - p^*||_2 \geq \frac{\Delta(P)}{2\e} \Rightarrow ||q' - p^*|| > diam(P) \Rightarrow q' \notin P
$.
We distinguish two cases when $q \in \Pe$ and $q \in \{\RR^d \mid q \notin P \ \ \wedge \ \ d(q, \partial P) > \e \cdot diam(P)\}$.
 
\noindent -- Let $q \in \Pe$, we wish to select an $\e '$ for the ANN problem such that:
\begin{equation}
(1+\e ') < ||p_i - q||_2/||p^* - q||_2 \label{EconditionNN} 
\end{equation}
Essentially, this would imply that $p^*$ is the nearest neighbor of $q$, while every $p_i \in S \setminus \{p^*\}$ is not an $\e '$-NN of $q$.

Let $r_i = d(p^*, H_i) \geq \e \cdot diam(P)$, where $H_i$  is the hyperplane defining facet $F_i$. By construction, $d(p^*, H_i) = d(p_i, H_i)$. It follows that the segment $p^*p_i$ has length $2 r_i$, as it is perpendicular to $H_i$. 

\noindent Next, we define the projection of $q$ on the line spanned by the segment $p^*p_i$ as 
$ q_i = (p_i-p^*)\cdot q/||p_i-p^*||_2 $
and its distance from $H_i$ as $ a_i = d\left(q_i, H_i\right) \geq \e \cdot diam(P)$

\noindent Obviously now, as depicted in Fig.~\ref{figApmProof}:
\begin{align*}
||p_i - q_i||_2 &= r_i + a_i, \; ||p^* - q_i||_2 = r_i - a_i
\end{align*}
Therefore,
\begin{align*}
||p_i - q||_2^2 &= ||p_i - q_i||_2^2 + ||q - q_i||_2^2 = (r_i + a_i)^2 + k_i^2\\
||p^* - q||_2^2 &= ||p^* - q_i||_2^2 + ||q - q_i||_2^2 = (r_i - a_i)^2 + k_i^2,
\end{align*}
where $k_i = ||q - q_i||_2^2 < diam(P)$.
It follows that,
\begin{align*}
\frac{||p_i - q||_2^2}{||p^* - q||_2^2} = \frac{(r_i + a_i)^2 + k_i^2}{(r_i - a_i)^2 + k_i^2} = 1 + \frac{4r_i a_i}{(r_i-a_i)^2+k_i^2} \geq \\
\geq 1 + \frac{4\e^2 (diam(P))^2}{(r_i-a_i)^2+k_i^2} \geq 1 + \frac{4\e^2 (diam(P))^2}{2 (diam(P))^2} \geq 1 + 2\e^2
\end{align*}

Substituting in \eqref{EconditionNN}, yields:
$
(1+\e') < \sqrt{1+2\e^2}  \Rightarrow \e' < \sqrt{1+2\e^2} - 1
$.

\smallskip
\noindent -- Let $q \in \{\RR^d \mid q \notin P \ \ \wedge \ \ d(q, \partial P) > \e \cdot diam(P)\}$. 
Assume the nearest neighbor of $q$ is $p_i \in S \setminus \{p^*\}$. Similarly, we are looking for an $\e'$ such that:
$$ (1 + \e') < ||p^*-q||_2/||p_i-q||_2 $$
\noindent This means $p^*$ cannot be an ANN of $q$. Now, like before:
\begin{align*}
\frac{||p^*-q||^2_2}{||p_i-q||^2_2} = \frac{(r_i + a_i)^2 + k_i^2}{(r_i - a_i)^2 + k_i^2} = 1 + \frac{4r_i a_i}{(r_i-a_i)^2+k_i^2} \geq \\ 
\geq 1 + \frac{4(\e \cdot diam(P))^2}{(r_i-a_i)^2+k_i^2} \geq 1 + \frac{4(\e \cdot diam(P))^2}{2 \left(\frac{2 \Delta(P)}{2\e}\right)^2} \geq \\ \geq 1 + \frac{4 \e^4 \cdot diam^2(P)}{2 \Delta^2(P)} > 1 + \frac{4 \e^4 \cdot diam^2(P)}{4 \cdot diam(P)} > \\ > 1 +  e^4 \cdot diam(P)
\end{align*}
It follows that, $ \e' < \sqrt{ e^4 \cdot diam(P)} - 1 $.

\noindent Choosing $\e' = \min \{\sqrt{e^4 \cdot diam(P)} - 1, \sqrt{1+2\e^2} -1\}$ and answering $\e'$-ANN queries on this set solves the original problem, because if a query point $q \in \Pe$, then we have ensured that the $\e'$-ANN data structure will correctly identify $p^*$ as the only approximate nearest neighbor of $q$. Similarly in a symmetric argument, for every $q \notin P$, such that $d(q, \partial P) > \e \cdot diam(P)$, $p^*$ will not be an approximate nearest neighbor of $q$. Lastly, if $d(q, \partial P) \leq \e \cdot diam(P)$ the response from the ANN data structure does not matter. Therefore, the reduction is complete. \qed
\end{proof}

\begin{figure}[ht]\centering
\includegraphics[scale=0.81]{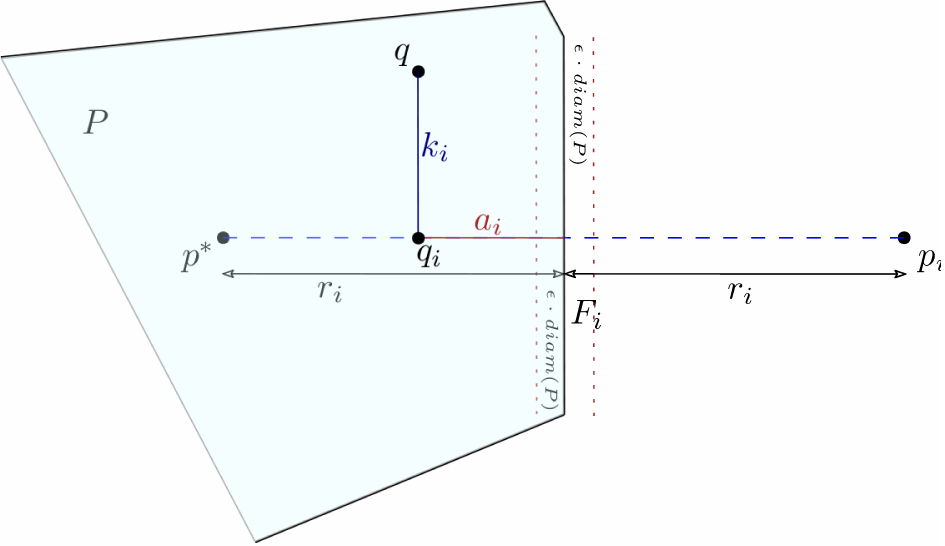}
\caption{$p_i$ corresponds to the symmetric point of $p^*$ about the facet $F_i$. We decompose the distances $||p^*-q||_2$ and $||p_i-q||_2$ and express them in terms of $a_i$ and $k_i$. 
Notice how $q \in \Pe \Rightarrow a_i \geq \e \cdot diam(P)$ and how $k_i < diam(P)$, as $q$ cannot be a vertex.
\label{figApmProof}}
\end{figure}

We now employ approaches for high-dimensional ANN to obtain a polynomial bound on the dimension by introducing a probability of success. Below, $\tilde{O}$ omits logarithmic factors.

\begin{theorem}\label{mainTheorem}[AMO in High Dimension]
For an H-polytope $P \subset \RR^d$ and an approximation parameter $\e$, such that $\Pe \neq \emptyset$, we can solve the Approximate Polytope membership problem on $P$ by building a data structure on $P$ answering queries in $\tilde{O}(d n^{\rho+o(1)})$ time and using $\tilde{O}(n^{1+\rho+o(1)} + d n)$ space, with a high probability of success, where $\rho = 1/(2 (1+\e')^2 -1)$ and $ \e' = \min \{\sqrt{e^4 \cdot diam(P)} - 1, \sqrt{1+2\e^2} -1\}$.
\end{theorem}

\begin{proof}
The Chebyshev center of a polytope $P$ is the center of the largest inscribed ball. Formally:
$ \arg \min\limits_{x \in P} \max\limits_{y \in P} ||x - y||_2^2 $.
Let $c$ be the Chebyshev center of $P$ with radius $r$ and assume $c \notin \Pe$, in order to deduce an absurdity. 
\begin{equation}\label{eRadiusCheb}
c \notin \Pe \Rightarrow r<\e \cdot diam(P)
\end{equation}
Take a point $c' \in \Pe$, as $\Pe \neq \emptyset$. 
\begin{equation}\label{eNewRadiusCheb}
d(c', F_i) \geq \e \cdot diam(P), \ \ 1 \leq i \leq n \Rightarrow B(c', \e \cdot diam(P)) \subset P
\end{equation}
Combining \eqref{eRadiusCheb} and \eqref{eNewRadiusCheb} produces an absurdity as we have found a larger inscribed ball in $P$, contradicting the property of $c$. Therefore, $c \in \Pe$.
We use $p^*=c$ as the starting point of the construction of the pointset $S$ in the proof of Theorem~\ref{lemmaAPM}. Answering ANN queries on $S$ using the LSH data structure of \cite{AndRaz15}, completes this proof. \qed
\end{proof}

\noindent \textit{Remark. } Any high-dimensional ANN solution can be utilized in the last step of Theorem~3 and we can inherit its complexity and its properties.

\section{Application to Polytope Boundary Problem}\label{Sbound}

The polytope boundary problem consists of creating a data structure for an H-polytope $P$ such that, given a query ray emanating from inside the polytope, we can efficiently compute the point $p=r \cap \partial P$. It is possible to achieve query time in $O(\log n)$ by using space in $O(n^d/\log^{\lfloor d/2 \rfloor}n)$~\cite{Ramos99}. The boundary oracle is dual to finding the extreme point in a given direction among a known pointset.
This is $\epsilon$-approximated through $\epsilon$-coresets for measuring extent, in particular (directional) width, but requires a subset of $O((1/ \epsilon)^{(d-1)/2})$ points~\cite{Agarwal05}.
The exponential dependence on $d$ or the linear dependence on $n$ make these methods of little practical use in high dimensions. 
Ray shooting has been studied in practice only in low dimensions, as well.

\begin{figure}\centering
\includegraphics[scale=0.25]{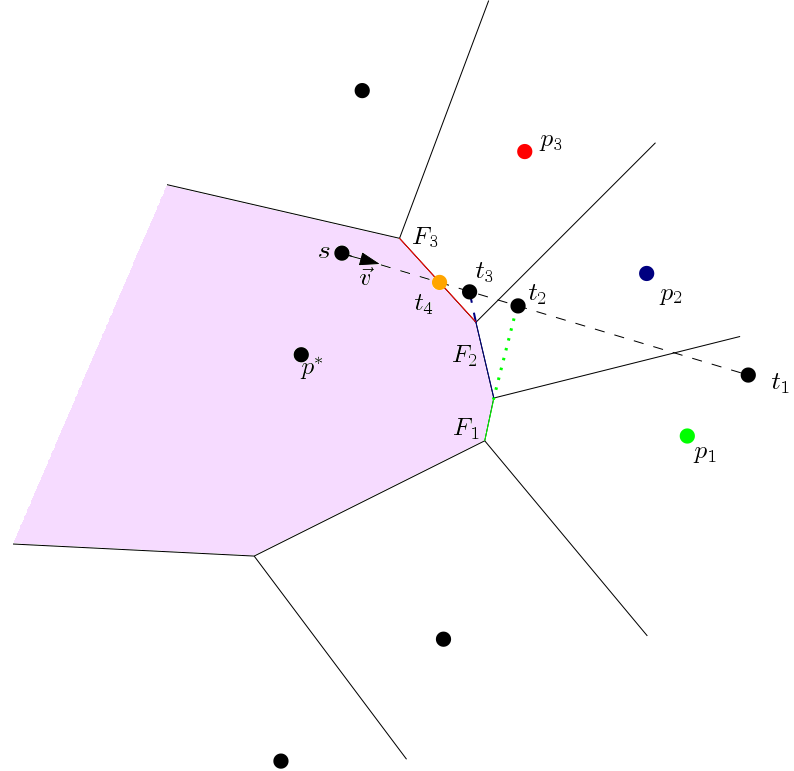}
\caption{An example of the boundary oracle converging to a solution. The query ray is $r=(s, \vec{v})$ and $t_4 = r \cap \partial P$ is the solution. $t_1, t_2, t_3, t_4$ were computed in sequence.} \label{figBoundProof}
\end{figure}

\paragraph{Exact Polytope Boundary Oracle.}
We now describe an iterative procedure for $P$ based on an exact nearest neighbor data structure \texttt{E\_MEM} defined on the pointset $S$ of Corollary~\ref{Cpoints} that we described in section~\ref{sAPMO}. This exact nearest neighbor data structure will act as the exact membership oracle for the polytope $P$.  We call this algorithm \texttt{BoundaryOracle}.

\textit{Finding the starting point.} 
The first step is to find a starting point $t_1$ such that $t_1\in r$ and $t_1\notin P$. 
We may use the intersection of $r$ with a bounding box around $P$.
A bounding box of $P$ can be readily computed by solving $2d$ linear programs to compute the farthest points on $P$ along the coordinate directions. 

\textit{Finding the intersection point.} We obtain an efficient method following a derivative-like approach.
Given starting point $t_1\notin P$: let $p_i$ be the nearest neighbor of $t_1$ using the data structure defined for membership: $p_i = \texttt{E\_MEM}(t_1)$.
Let $H_i$ be the hyperplane supporting the facet $F_i$ used to define $p_i$;
$F_i$ separates the cell of $p_i$ from $P$ in the Voronoi diagram. 
Let $t_2= (H_i \cap r)$.
Iterate by computing $t_3,t_4,\ldots$, until membership decides $t_n \in P$.

\counterwithout{lstlisting}{chapter}

\begin{lem}[Correctness of algorithm \texttt{BoundaryOracle}]\label{lemBoundaryOracle}
\texttt{BoundaryOracle} always converges to a solution for the boundary problem for a given polytope $P$.
\end{lem}

\noindent The proof of this lemma is presented in the Appendix~\ref{appProofs}.

\paragraph{Approximate Polytope Boundary Oracle.}
Now, we define an approximate version of the polytope boundary problem.
\begin{dfn}[Approximate Polytope Boundary Problem]\label{dfnAPBO}
Given a convex H-polytope $P \subset \mathbb{R}^d$ and an approximation 
parameter $\e\in (0,1)$, preprocess $P$ into a data structure such that, given 
a query ray $r \subset \mathbb{R}^d$ emanating from inside $P$, it is possible to efficiently compute a point
$r^* \in r$ such that $d(r^*, \partial P) \leq \e \cdot diam(P)$.
\end{dfn}

We make two additional changes to the algorithm presented in the previous section. First, we compare $t_i$'s and $t_{i+1}$'s distance from the ray's source point $s$. If the distance is not improved, then we discard the current $t_{i+1}$ and set it as $t_{i+1} = (t_{i} - s) - \frac{v}{||v||_2}\e$. In other words, in this case we take an $\e$-step from $t_{i}$ towards the ray's apex.
The second change concerns termination. Now we stop when the approximate membership oracle identifies a point $t_i$ as being inside the polytope, or when the point $t_i$ lies in the opposite direction of the ray.

\begin{lstlisting}[caption={Approximate Boundary Oracle},label={aAppxBoundary}]
Input: $H$-polytope $P \subset \RR^d$, ray $r$ (pair $(s, v)$), $\e$
Output: $t \in \RR^d$ s.t. $t \in r$ and $d(t, \partial P ) \leq \e diam(P)$

A_MEM = approximate membership oracle for $P$
$Q$ = bounding_box($P$)
$t = Q \cap r$;
do
  $p_i$ = A_MEM(t);
  if $p_i$==$p∗$ then return $t + \frac{v}{||v||_2} \e$; end
  $t_{prev} = t$
  $H = H_i$ //facet corresponding to pi
  $t = H \cap r$
  if $||t - s||_2 \geq ||t_{prev} - s||_2$ then $t = (t_{prev} - s) - \frac{v}{||v||_2} \epsilon$; end
  if $(t - s) \cdot v < 0$ then return $s + \frac{v}{||v||_2} \e$; end
while True;
\end{lstlisting}
%


\begin{lem}[Correctness of Algorithm~\ref{aAppxBoundary}]\label{lemAppxBoundaryOracle}Algorithm~\ref{aAppxBoundary} always converges to a solution for the approximate boundary problem.
\end{lem}

\noindent We present the proof of this lemma in the Appendix~\ref{appProofs}.

\section{Implementation and Experiments}\label{Simplement}

\noindent \textit{Implementation.} All of our code\footnote{\url{https://github.com/van51/volume_approximation}} is linked to the software of~\cite{EmiFis14}. It is written in C++11 based on using the CGAL\footnote{\url{http://www.cgal.org/}} library for the readily available data structures of d-dimensional objects, Eigen3 for some linear algebra computations and FALCONN\cite{And15} for the approximate nearest neighbor data structure. We remind the reader at this point that for a polytope $P(d,n,i)$ we compute $n+1$ points, out of which one point $p^* \in P$ while all remaining $n$ points $p_i \notin P, 1 \leq i \leq n$.
FALCONN offers LSH only for angular distances so in order to take advantage of that we use it in the following manner. We consider our pointset already centered around the internal point, in our case the origin. We build a FALCONN data structure using the Hyperplane LSH family and setting $k=11, l=1$, number of probes=$40$, when the number of facets $n \geq 10000$. Otherwise, we set them to $l=1$, $k=8$ and number of probes=$150$. $l$ corresponds to the number of hash tables built, $k$ corresponds to the number of hash functions used per hash table and number of probes is a parameter for the multi-probe LSH scheme~\cite{LV07}. The data structure is built for every computed point besides the internal one. Then, assuming that for a query $q$ FALCONN returns an approximate nearest neighbor guess $x_i$, we compare $d(x_i, q)$ to $d(p^*, q)$ and return the point closest to $q$ out of $x_i, p^*$. The parameters for FALCONN were selected manually, while trying to maintain a 90\% success rate for membership.  

\textit{Datasets. } We experiment on a synthetic dataset consisting of high-dimensional polytopes with a large number of facets. In particular,
for the following set of possible dimensions $\boldsymbol{d} = \{ 40, 100, 500, 1000 \}$ and the following set of possible number of facets $\boldsymbol{n} = \{5000, 10000, 20000,$ $ 50000, 100000, 500000, 1000000\}$, we generate $5$ polytopes for every combination of $\boldsymbol{d} \times \boldsymbol{n}$. Each polytope $P(d,n,i), d \in \boldsymbol {d}, n \in \boldsymbol{n}, i \in \{1,2,3,4,5\}$ lives in a $d$-dimensional Euclidean space and is described by $n$ inequalities of the form:
$ a_j x \leq 1000, 1 \leq j \leq n, $ where $a_j \sim mod(U(0, 32767), 1000)$. The notation $U(i,j)$ denotes the uniform real distribution over $[i, j]$. By construction, each polytope contains the origin $0$, which we use as the internal point needed by the approximate membership oracle. If that assumption was not satisfied, we could have computed an internal point either by solving a linear program or by computing an important point of the polytope, like the Chebyshev center. 

\textit{Evaluation protocol.} For both oracles we report pre-processing time, total query time, and success rate vs $n$ and $d$ as $n$ and $d$ vary in their respective sets $\boldsymbol{n},\boldsymbol{d}$. Specifically for the boundary oracle we also report the average number of steps that it required in order to reach a solution and we also compute the min,max and average distances of the point returned from our approximate boundary oracle to the actual point that the exact ray shooting problem should have computed. We compare the query time to the naive approach of checking all $n$ facets of $P$. For the membership oracle we sample $1000$ query points inside the polytope via the popular hit-and-run paradigm and then move these points sufficiently far from the origin so that they lie outside the polytope. This generates another $1000$ points to form a total of $2000$ points. Similarly for the boundary oracle we use $1000$ query points in total.

\textit{Results.} Table~\ref{tablePreprocessing} depicts the total time in seconds for creating the approximate membership oracle on random polytopes for different values of $d, n$.  
Figure~\ref{fmr} depicts total time in seconds for all queries to be completed. Parameters were tuned such that the membership oracle achieved an accuracy of $>90\%$, i.e. at least $9$ out of $10$ queries succeed on average. 
The results matched our expectations with regards to the behaviour of the oracles in high dimension, where we can see a huge difference in the query time, especially as the number of facets grows larger as well.

\begin{table}
\centering
\caption{Preprocessing time in seconds for membership oracle. This includes computing the $n+1$ pointset and creating the ANN data structure on top of it.}
\begin{tabular}{cc|c|c|c|c|c|c|c|}
\cline{3-9}
                                                          &               & \multicolumn{7}{c|}{\textbf{Number of facets}}                                                                          \\ \cline{3-9} 
                                                          &               & \textbf{5000} & \textbf{10000} & \textbf{20000} & \textbf{50000} & \textbf{100000} & \textbf{500000} & \textbf{1000000} \\ \hline
\multicolumn{1}{|c|}{\multirow{4}{*}{\textbf{Dimension}}} & \textbf{40}   & 0.006s        & 0.013s         & 0.027s         & 0.057s         & 0.125s          & 0.518s          & 0.795s           \\ \cline{2-9} 
\multicolumn{1}{|c|}{}                                    & \textbf{100}  & 0.015s        & 0.035s         & 0.057s         & 0.121s         & 0.230s          & 1.005s          & 1.885s           \\ \cline{2-9} 
\multicolumn{1}{|c|}{}                                    & \textbf{500}  & 0.055s        & 0.108s         & 0.193s         & 0.419s         & 0.717s          & 3.396s          & 6.744s           \\ \cline{2-9} 
\multicolumn{1}{|c|}{}                                    & \textbf{1000} & 0.101s        & 0.192s         & 0.342s         & 0.783s         & 1.470s          & 5.500s          & 10.770s          \\ \hline
\end{tabular}

\label{tablePreprocessing}
\end{table}

\begin{figure}\centering
\includegraphics[scale=0.23]{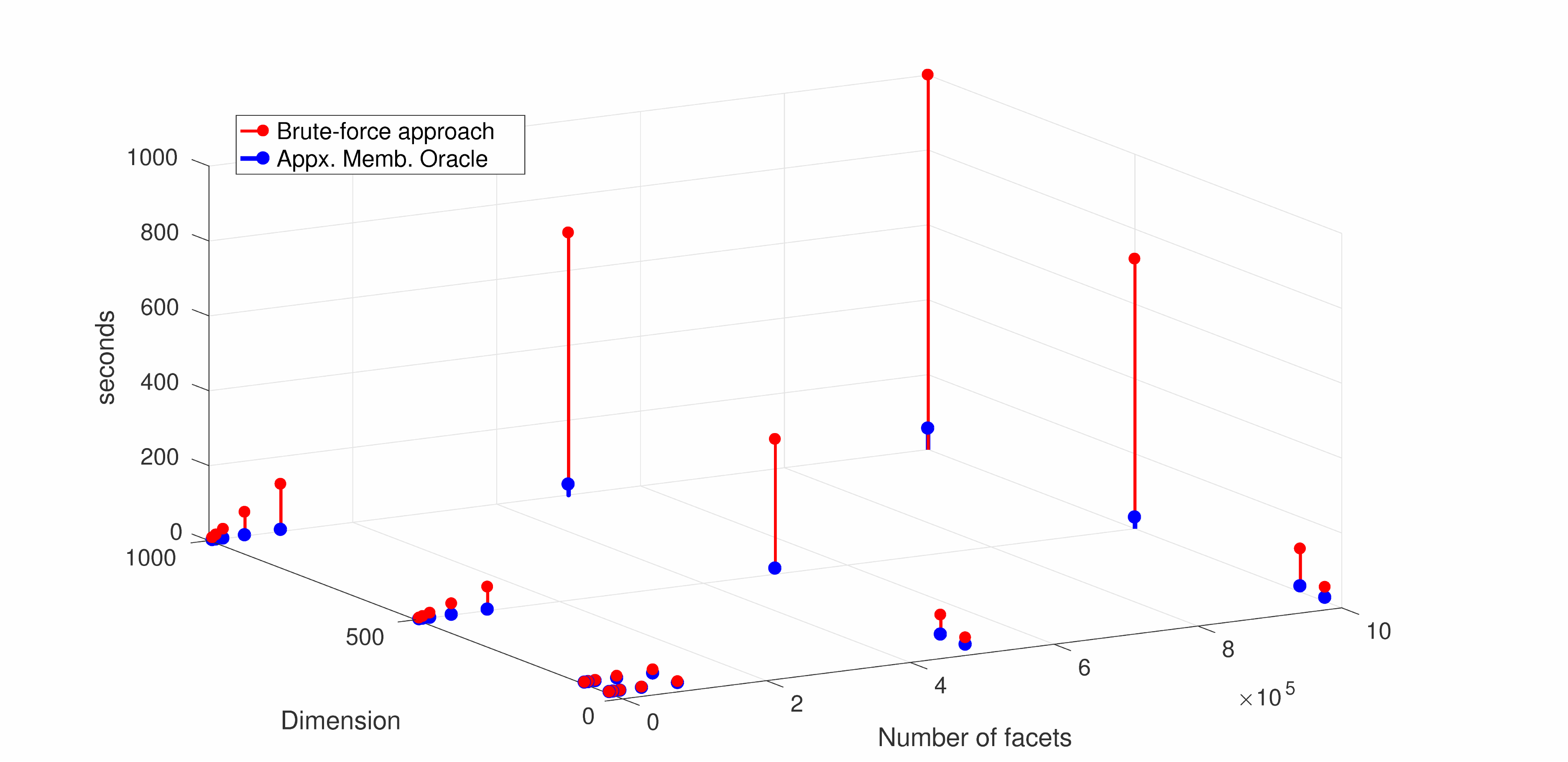}
\caption{Average timing results for 2000 queries for varying $n$ and $d$. Half of the queries were inside the random polytopes and half were outside.}
\label{fmr}
\end{figure}


\bibliographystyle{splncs}
\bibliography{biblio}
\appendix
\section{Appendix}

\subsection{Proofs of section~\ref{Sbound}}\label{appProofs}
Proof of lemma~\ref{lemBoundaryOracle}.

\begin{proof}
Let $t_1, t_2, \ldots $ denote the sequence of successive points computed on the ray $r$ by the above algorithm. Let $x_1, x_2, \ldots $ be a sequence of points in $S$, each representing the nearest neighbor of the point $t_i$. This means that the nearest neighbor of $t_i$ was $x_i$ at the $i$-th step. We assume without loss of generality that each $t_i$ has a single nearest neighbor, because otherwise it would mean that $t_i$ falls on the intersection of a line (the ray), a Voronoi facet and a supporting hyperplane which is highly degenerate. However, even in that case we could consider every nearest neighbor of the point and take the one that improves the distance the most.
For correctness, assume that we have reached the $i$-th step. There are two cases for $t_{i+1}$. Either it lies on $\partial P$ in which case the membership data structure $\mathtt{E\_MEM}$ will identify it as being inside and the algorithm will terminate. Otherwise, by convexity of the cell of $x_i$, $t_{i+1}$ lies between $\partial P$ and $t_{i}$, since $t_{i+1}$ lies on an ``extension'' of the facet (meaning on $H_i \setminus F_i$) between the cell of $x_i$ and $P$. Since $H_i \setminus F_i$ cannot belong to a Voronoi facet, $t_{i+1}$ will always belong to a new Voronoi cell. Therefore the sequence $x_i$ will not have any repeating points and the algorithm will eventually reach $\partial P$ where the iteration will stop and return $\partial P \cap r$.
\end{proof}

\noindent Proof of lemma~\ref{lemAppxBoundaryOracle}.
\begin{proof}
Observe that the successive points $t_i$ lying on the ray $r$ are always improving the distance to the ray's apex, by a factor of at least $\e$. Additionally, by definition, the ray's apex always lies inside $P$. We separate two cases for the ray's apex, which we will from now on denote as $s$. 
\begin{enumerate}
\item $d(s, \partial P) > \e \cdot diam(P) + \e$
\item $d(s, \partial P) \leq \e \cdot diam(P) + \e$
\end{enumerate}
In case 1, the algorithm will eventually reach a point $t_i$, after performing a number of $\e$-steps, such that $t_i \in P$ and $d(t_i, \partial P) \geq \e \cdot diam(P)$. Since the ray's apex $s$ is at distance $> \e \cdot diam(P) + \e$ from $\partial P$ this will happen while $(t_i - s) \cdot v > 0 $. In this case we return point $t_i + \frac{v}{||v||_2}\e$ which lies within distance $\e \cdot diam(P)$ from $\partial P$.

In case 2, the point $t_i$ will either reach $d(t_i, \partial P) > \e \cdot diam(P)$ and will be identified as being inside and in which case the algorithm will correctly return point $t_i + \frac{v}{||v||_2}\e$. Alternatively, it will take an $\e$-step and move to the opposite direction of the ray. In that case, $s$ is identified as lying at distance at most $\e diam(P) + \e$ from $\partial P$ and in which case we return point $s + \frac{v}{||v||_2}\e$ which lies in $r$ at distance $< \e \cdot diam(P)$ from $\partial P$.

Eventually, the algorithm returns point $t$: $t \in P$ and $d(t, \partial P) \geq \epsilon \cdot diam(P)$.
\end{proof}

\end{document}